\documentclass[12pt]{article}

\usepackage[utf8]{inputenc} 
\usepackage[T1]{fontenc}
\usepackage{mathpazo} 

\usepackage[english]{babel}
\usepackage{float}
\usepackage{placeins}
\usepackage{titlesec}
\usepackage[titletoc]{appendix}
\titleformat{\section}{\normalfont\large\bfseries\centering}{\thesection}{1em}{}
\titleformat{\subsection}{\normalfont\normalsize\bfseries\centering}{\thesubsection}{1em}{}
\titleformat{\subsubsection}{\normalfont\normalsize\itshape}{\thesubsubsection}{1em}{}
\usepackage{graphicx}
\usepackage{subcaption, doi}
\usepackage{caption}
\usepackage{tikz}
\usetikzlibrary{positioning,shapes.geometric,arrows.meta}
\usepackage{pgfplots}
\pgfplotsset{compat=1.17}
\usepackage{rotating}
\usepackage{caption}
\usepackage{ragged2e} 

\usepackage{tabularx, ragged2e}
\usepackage{array}
\usepackage{multirow}
\usepackage{booktabs, longtable}
\usepackage{adjustbox}
\usepackage{diagbox}

\usepackage{amsmath, amssymb, amsfonts, mathtools, apacite}
\usepackage{siunitx}
\usepackage{amsthm}

\theoremstyle{plain}

\newtheorem{proposition}{Proposition}
\newtheorem{lemma}{Lemma}
\newtheorem{corollary}{Corollary}
\newtheorem{assumption}{Assumption}

\theoremstyle{definition}
\newtheorem{definition}{Definition}

\theoremstyle{remark}
\newtheorem{remark}{Remark}
\newtheorem{example}{Example}
\usepackage[authoryear]{natbib}
\usepackage{breakcites}
\usepackage{doi}

\usepackage{geometry}
\usepackage{setspace}
\usepackage{lscape}
\usepackage{pdflscape}
\usepackage{changepage}
\usepackage[left,mathlines]{lineno}
\usepackage{microtype}
\usepackage{hyphenat}
\usepackage{titlesec}
\usepackage{epigraph}
\usepackage{appendix}
\usepackage{xcolor}
\usepackage{authblk}

\definecolor{maroon}{rgb}{0.5, 0.0, 0.0}
\definecolor{oxfordblue}{rgb}{0.0, 0.13, 0.28}

\newcommand{\E}{\mathbb{E}}

\newcommand{\Prob}{\mathbb{P}}

\usepackage{hyperref}
\hypersetup{
	colorlinks=true,
	linkcolor=blue,
	citecolor=blue,
	filecolor=magenta,
	urlcolor=blue
}

\graphicspath{{Figures/}}
\DeclareUnicodeCharacter{21B5}{$\hookleftarrow$}

\title{Information-Credible Stability in Matching with Incomplete Information}
\author{\textsc{Kaibalyapati Mishra}}
\affil{\textsc{CESP, ISEC}, \textsc{Bengaluru, India}}
\date{\today}

\begin{document}
	
	\maketitle
	
	\doublespacing
	
	\begin{abstract}
		
		In this paper, I develop a refinement of stability for matching markets with incomplete information. I introduce\textit{\textbf{ Information-Credible Pairwise Stability (ICPS)}}, a solution concept in which deviating pairs can use credible, costly tests to reveal match-relevant information before deciding whether to block. By leveraging the option value of information, ICPS strictly refines Bayesian stability, rules out \textit{fear-driven matchings}, and connects belief-based and information-based notions of stability.	ICPS collapses to Bayesian stability when testing is uninformative or infeasible and coincides with complete-information stability when testing is perfect and free. I show that any ICPS-blocking deviation strictly increases total expected surplus, ensuring welfare improvement. I also prove that ICPS-stable allocations always exist, promote positive assortative matching, and are unique when test power is sufficiently strong. The framework extends to settings with non-transferable utility, correlated types, and endogenous or sequential testing. 
		
	\end{abstract}


\pagebreak

\section{INTRODUCTION}

The classical framework of matching, following \citet{gale1962college} and \citet{roth1989two}, rests on the assumption of complete information. In this benchmark, stability is a clear and mechanical consequence of mutual optimization: given full transparency, no (or more) profitable deviation exists. The subsequent shift to \emph{incomplete information} was a necessary step toward realism, acknowledging that agents in labor markets, school choice, and other domains must often act on beliefs rather than facts \citep{bikhchandani2014two, liu2014stable}.

However, this shift introduced a subtle but profound conceptual tension. As \citet{liu2014stable} elegantly demonstrate (or a conservative interpretation would), stability in such settings can emerge not from a lack of mutually beneficial re-matching opportunities, but from agents' \emph{reluctance to deviate in the face of uncertainty}. This \emph{fear-driven matching} arises when the potential downside of a deviation—ending up with a partner even worse than one's current match—outweighs the uncertain upside. The market settles, but into an equilibrium sustained by risk aversion and informational opacity.

This paper is motivated by a simple observation: real-world agents are not passive prisoners of their priors. When contemplating a deviation—a job switch, a new hire, a different school—they do not simply calculate expected values and remain paralyzed by fear. Instead, they actively seek to \emph{reduce} the very uncertainty that constrains them. They conduct interviews, consult networks, run trials, or request disclosures. A stability concept that ignores this fundamental aspect of strategic behavior risks conflating a truly robust outcome with one that is merely an artifact of assumed informational helplessness.

I therefore challenge the notion that stability under incomplete information must inherently be a story of aversion. The \emph{Information-Credible Pairwise Stability} (ICPS) concept proposed here reframes stability as a property of outcomes that can withstand not just the \emph{desire} to deviate, but the \emph{informed capability} to do so. It asks:\textit{ would this match survive if the agents involved were given a credible opportunity to learn about each other?} By incorporating the option value of information directly into the blocking condition, ICPS achieves stability through a process of \emph{mutually improving recontracting}, made possible by bilateral information acquisition, rather than through deterrence born of ignorance.

This perspective bridges a gap between the theoretical purity of the complete-information benchmark and the messy reality of asymmetric information. It suggests that the efficiency of a matching market is not solely determined by its formal rules, but also by the cost and credibility of the information channels available to its participants. A market with cheap, reliable verification mechanisms will naturally converge toward the efficient, complete-information outcome, while a market where information is prohibitively costly or unreliable may indeed be characterized by the fragile, fear-driven stability of the Bayesian benchmark. The ICPS framework provides the language and tools to formally analyze this critical distinction.


Our work builds on and extends several strands of the matching literature. \citet{bikhchandani2014two} introduce ex ante and Bayesian stability in one-sided incomplete information environments, while \citet{liu2014stable} extend these concepts to two-sided settings, establishing how uncertainty alone can sustain stable outcomes. \citet{chen2023theory} develop the notion of ``information stability,'' focusing on the effects of information on outcomes but without integrating credible information acquisition into the blocking mechanism. \citet{alston2020non} highlight the non-existence of stable matchings in certain incomplete information settings.

A key insight of this literature is that incomplete information can generate what \citet{liu2014stable} describe as \emph{fear-driven stability}, where agents rationally refrain from blocking even inefficient allocations because uncertainty about partners’ types makes deviation risky. This mechanism allows stable outcomes to emerge even when mutual gains from deviation exist under full information. \citet{bikhchandani2014two} formalize both ex ante and Bayesian notions of stability and show how belief structures shape the set of stable outcomes. These belief-based concepts are elegant but restrictive: they hinge on priors, which may be unrealistic, outdated, or uninformative in many real-world settings. As a result, stability is determined less by fundamentals and more by informational opacity.

Subsequent work examines how this informational friction shapes the \emph{structure and efficiency} of stable matchings. \citet{peralta2025lone} extend the lone wolf theorem to incomplete information settings with transfers, showing that some structural properties of stable matchings persist despite informational uncertainty. \citet{feng2021matching} explore the link between Bayesian stability and Bayesian efficiency, demonstrating that under specific belief assumptions, stability and efficiency can coincide. Relatedly, \citet{peralta2024not} show that while incomplete information weakens the strong positive assortativity of classical matching, sorting properties can be partially recovered under restrictions on stabilizing sets. These papers highlight that belief-driven stability, though permissive, still embeds some structural regularities.

Another important branch of the literature studies settings with \emph{two-sided incomplete information}, where both firms and workers have private information. \citet{chen2023theory} introduce a theory of information stability, embedding the information structure into the stability definition and showing existence results under two-sided uncertainty. \citet{park2025stable} study two-sided asymmetric information and identify conditions under which assortativity and efficiency can still be achieved. Together, these contributions mark a conceptual shift: instead of purely belief-based blocking, they begin to incorporate informational structure directly into the stability condition.

Our work contributes to this evolving conversation by introducing a new refinement of stability, \emph{Information-Credible Pairwise Stability (ICPS)}, which explicitly embeds credible bilateral information acquisition into the blocking mechanism. Unlike existing belief-based stability concepts, ICPS shifts the focus from \emph{deterrence through uncertainty} to \emph{stability through verifiable opportunity}. By allowing deviating pairs to acquire and act upon credible information, ICPS provides a more realistic foundation for stability in information-rich matching environments.

Our contribution is twofold. \textit{First}, I endogenize information acquisition and incorporate \emph{credible bilateral testing} directly into the definition of blocking, thereby linking incentives to the underlying information structure rather than fixed priors. \textit{Second}, I position ICPS in relation to existing stability notions and show that it (i) strictly refines Bayesian stability, (ii) collapses to Bayesian stability when credible testing is infeasible or uninformative, and (iii) coincides with complete-information stability when testing is perfect and costless. In doing so, ICPS bridges the earlier Bayesian stability literature \citep{liu2014stable,bikhchandani2014two} with more recent work on information structures \citep{chen2023theory}, while also connecting to the welfare and sorting implications emphasized by \citet{peralta2025lone}, \citet{peralta2025lone}, and \citet{feng2021matching}. I build on the classical complete-information benchmark of \citet{roth1989two}, against which both Bayesian and ICPS stability can be rigorously compared.

The rest of the paper is organised as follows. The next section~\ref{sec:env} defines the environment, Section~\ref{sec:mainresults} presents the main results, extensions are discussed in Section~\ref{sec:extensions}, while Section~\ref{sec:con} concludes.

\section{ENVIRONMENT} \label{sec:env}

Let’s begin with the standard environment set-up procedure. I consider a two-sided matching environment with a finite set of firms $F$ and a finite set of workers $L$. Each firm $f \in F$ and worker $\ell \in L$ has a private productivity type $\theta_i \in \{H,L\}$, interpreted as \emph{high} or \emph{low} productivity. Types are drawn independently with $\Pr(\theta_i = H) = p \in (0,1)$.

The total match surplus depends on the pair of realized types of the agents. For any firm–worker pair $(f,\ell)$:
\begin{equation}
	S(\theta_f,\theta_\ell) =
	\begin{cases}
		\alpha & \text{if } (\theta_f,\theta_\ell) = (H,H) \\
		\beta  & \text{if } (\theta_f,\theta_\ell) = (H,L) \text{ or } (L,H) \\
		\gamma & \text{if } (\theta_f,\theta_\ell) = (L,L),
	\end{cases}
\end{equation}
where $\alpha > \beta > \gamma \geq 0$. An unmatched agent obtains a payoff of $0$. The surplus function is strictly increasing in both arguments with the conventional wisdom.

A \emph{matching} is a function $\mu: F \to L \cup \{\emptyset\}$ that assigns each firm to at most one worker (and vice versa), where $\mu(f) = \emptyset$ indicates $f$ is unmatched. The payoff to a firm–worker pair depends on their realized types and any transfers between them.

An \emph{allocation} consists of a matching $\mu$ and a transfer vector $p$, where $p_{f,\ell}$ denotes the transfer between $f$ and $\ell$ if matched. Total surplus under allocation $(\mu,p)$ is
\begin{equation}
	W(\mu,p) = \sum_{f \in F} S(\theta_f, \theta_{\mu(f)}),
\end{equation}
with the convention that $S(\theta_f, \theta_{\emptyset}) = 0$.

Agents observe their own type but not their partner’s. Prior beliefs regarding agent types are common knowledge. Under Bayesian stability, no additional information is available before deviation. Under ICPS, a deviating firm–worker pair $(f,\ell)$ may access a \emph{credible bilateral test} that reveals the true types of both agents with accuracy $\pi \in [0,1]$ at cost $c \geq 0$. This credible information influences the expected gains from deviation. In entry-level labor markets, the corresponding test is analogous to interviews, while in higher education markets, it is typically represented by standardized tests and other screening instruments.

\begin{assumption}[Surplus Monotonicity]\label{ass:mono}
	S(H,H) > S(H,L) = S(L,H) > S(L,L).
\end{assumption}

\begin{assumption}[Independent Types]\label{ass:indep}
	Types are independently distributed across firms and workers.
\end{assumption}

\begin{assumption}[Positive Surplus]\label{ass:pos}
	$S(\theta_f,\theta_\ell) \ge 0$ for all $(\theta_f,\theta_\ell)$.
\end{assumption}

\begin{assumption}[Test Power for Strict Blocking]\label{ass:testpower}
	Let 
	\[
	\Delta \;\equiv\; \min\big\{\, S(x,y)\;-\;\max\{S(x,y'),\,S(x',y)\} \;:\; x>x',\ y>y' \big\}.
	\]
	Under Assumption~\ref{ass:mono} and finite type sets, $\Delta>0$. 
	The feasible test set $\mathcal{T}$ contains some $(\bar\pi,\bar c)$ such that 
	\[
	\bar\pi\,\Delta \;>\; \bar c.
	\]
\end{assumption}

\begin{remark}[Computing $\Delta$ in the H/L parametrization]
	With $S(H,H)=\alpha$, $S(H,L)=S(L,H)=\beta$, $S(L,L)=\gamma$ and $\alpha>\beta>\gamma$, 
	\[
	\Delta \;=\; \min\{\alpha-\beta,\ \beta-\gamma\} \;>\; 0.
	\]
	Assumption~\ref{ass:testpower} then requires a feasible test $(\bar\pi,\bar c)$ with $\bar\pi\min\{\alpha-\beta,\beta-\gamma\}>\bar c$.
\end{remark}

\subsection{ICPS Stability}

\subsubsection{Individual Rationality}

\begin{definition}[Individual Rationality]
	An allocation $(\mu,p)$ is individually rational if
	\begin{equation}
		\pi_f(\mu,p) \geq 0 \quad \text{and} \quad \pi_\ell(\mu,p) \geq 0 \quad \text{for all } f \in F, \ \ell \in L,
	\end{equation}
	where $\pi_f$ and $\pi_\ell$ denote the expected payoffs to the firm and worker, respectively.
\end{definition}

Let $\Sigma^0$ denote the set of all individually rational allocations.

\subsubsection{Bayesian Stability Benchmark}

Under Bayesian stability, no pair has access to additional information prior to deviating. A firm–worker pair $(f,\ell)$ can block $(\mu,p)$ if, given their priors, they can match and redistribute the surplus to make both weakly better off.

The expected surplus from matching a random firm and worker without further information is
\begin{equation}
	\mathbb{E}[S(\theta_f,\theta_\ell)] 
	= p^2 \alpha + 2p(1-p)\beta + (1-p)^2 \gamma.
\end{equation}

\begin{definition}[Bayesian Stability]\label{def:bayesian_stability}
	An allocation $(\mu,p)$ is \emph{Bayesian stable} if it is individually rational and 
	there exists no firm–worker pair $(f,\ell)$ and transfer $t$ such that
	\begin{equation}
		\mathbb{E}\big[S(\theta_f,\theta_\ell)\big] > \pi_f(\mu,p) + \pi_\ell(\mu,p),
	\end{equation}
	where the expectation is taken with respect to the common prior over types \emph{before types are realised}. 
\end{definition}

\begin{remark}[Ex Ante vs Interim Blocking]
	The stability concept adopted here follows the \emph{ex ante} formulation of \cite{liu2014stable}: 
	blocking deviations are evaluated using expected surplus given priors, before any agent observes their own type. 
	This differs from an \emph{interim} formulation, in which each agent observes their own type prior to considering deviation and expectations are conditional. 
	Our focus on ex ante stability aligns with the informational structure of ICPS, in which credible testing opportunities are introduced before type realizations affect blocking incentives.
\end{remark}

\subsubsection{Information-Credible Pairwise Stability}

Under ICPS, deviating pairs can acquire credible information before deciding to deviate. 
This can increase their expected deviation surplus, leading to a refinement of Bayesian stability.

\begin{definition}[Information-Credible Pairwise Stability (ICPS)]\label{def:ICPS}
	An allocation $(\mu,p)$ is \emph{ICPS-stable} if it is individually rational and there exists no firm–worker pair $(f,\ell)$ and feasible test $(\pi,c)$ such that
	\begin{equation}
		\mathbb{E}\big[S(\theta_f,\theta_\ell) \,\big|\, \text{deviation with test}\big] - c 
		> \pi_f(\mu,p) + \pi_\ell(\mu,p).
	\end{equation}
\end{definition}

\noindent
\textit{Clarification.} 
The term $\mathbb{E}[S(\theta_f,\theta_\ell)\mid \text{deviation with test}]$ refers to the 
expected joint surplus conditional on the pair’s testing strategy. 
Intuitively, the pair observes signals generated by the test and consummates the deviation only on favourable realizations, so this expectation is weighted by both signal structure and acceptance behavior. 
The formal decomposition of this term is provided in Lemma~\ref{lemma:surplus_decomposition}.

\begin{remark}[Null test as a feasible action]\label{rem:nulltest}
	Throughout, I allow deviating pairs to \emph{choose not to acquire information}. Formally, the action set contains a degenerate ``null test'' with accuracy $\pi=0$ and cost $c=0$, under which the pair receives no signal and decides to consummate deviation (or not) based on priors only. This is without loss, as any agent can refrain from testing in practice.
\end{remark}

\begin{proposition}[Refinement]\label{prop:refinement}
	For any feasible set of information acquisition technologies $(\pi,c)$:
	\begin{enumerate}
		\item ICPS implies Bayesian stability.
		\item If for all feasible technologies either $c=\infty$ or $\pi=0$, then ICPS coincides with Bayesian stability.
		\item If there exists a feasible technology with $c=0$ and $\pi=1$, then ICPS coincides with complete-information pairwise stability.
	\end{enumerate}
\end{proposition}

\begin{proof}[Proof of Proposition \ref{prop:refinement}]
	I prove each part in turn.
	
	\medskip
	\noindent\textbf{(1) ICPS $\Rightarrow$ Bayesian stability.}
	Suppose, toward a contradiction, that $(\mu,p)$ is ICPS-stable but not Bayesian stable. Then there exists a pair $(f,\ell)$ such that
	\[
	\mathbb{E}\big[S(\theta_f,\theta_\ell)\big] \;>\; \pi_f(\mu,p)+\pi_\ell(\mu,p).
	\]
	By Remark~\ref{rem:nulltest}, the pair can implement the null test $(\pi,c)=(0,0)$ and (by definition of ICPS) consider deviation with no additional information. Choosing the trivial acceptance strategy that consummates deviation regardless of the (uninformative) signal yields
	\[
	\mathbb{E}\big[S(\theta_f,\theta_\ell)\,\big|\, \text{deviation with test}\big]-c
	\;=\;\mathbb{E}\big[S(\theta_f,\theta_\ell)\big]-0
	\;>\; \pi_f(\mu,p)+\pi_\ell(\mu,p),
	\]
	which constitutes an ICPS blocking deviation by the pair $(f,\ell)$, contradicting ICPS-stability. Hence, ICPS implies Bayesian stability.
	
	\medskip
	\noindent\textbf{(2) ICPS $=$ Bayesian when testing is uninformative or infeasible.}
	Assume that for every feasible technology either $c=\infty$ or $\pi=0$.
	\begin{itemize}
		\item If $c=\infty$, no test can be undertaken; ICPS permits no information acquisition beyond priors, so the set of blocking deviations coincides with the Bayesian ones by definition.
		\item If $\pi=0$ for all feasible technologies, any test yields no information. For any acceptance strategy, the law of iterated expectations implies
		\[
		\mathbb{E}\big[S(\theta_f,\theta_\ell)\,\big|\, \text{deviation with test}\big]
		\;=\;\mathbb{E}\big[S(\theta_f,\theta_\ell)\big],
		\]
		so (weakly) optimal behavior is to forgo testing (or equivalently, incur zero expected gain from testing). Hence the ICPS blocking condition
		\[
		\mathbb{E}\big[S(\theta_f,\theta_\ell)\,\big|\, \text{deviation with test}\big]-c
		\;>\; \pi_f(\mu,p)+\pi_\ell(\mu,p)
		\]
		reduces to
		\[
		\mathbb{E}\big[S(\theta_f,\theta_\ell)\big]
		\;>\; \pi_f(\mu,p)+\pi_\ell(\mu,p),
		\]
		i.e., exactly the Bayesian blocking condition. Therefore, the two stability notions coincide.
	\end{itemize}
	
	\medskip
	\noindent\textbf{(3) ICPS $=$ complete-information stability when a free perfect test exists.}
	Suppose there exists a feasible technology with $\pi=1$ and $c=0$. Under this test, the deviating pair $(f,\ell)$ learns the realized types $(\theta_f,\theta_\ell)$ with certainty before deciding whether to consummate deviation. Let $V(\theta_f,\theta_\ell)$ denote the \emph{maximal joint surplus from deviation net of transfers} they can realize upon learning $(\theta_f,\theta_\ell)$; since transfers are unrestricted, $V(\theta_f,\theta_\ell)=S(\theta_f,\theta_\ell)$.
	
	Given current allocation $(\mu,p)$, let
	\[
	\Pi_{f\ell} \;\equiv\; \pi_f(\mu,p)+\pi_\ell(\mu,p)
	\]
	be the pair’s joint payoff under $(\mu,p)$ (these can be interpreted either as ex-ante expectations or, in the complete-information benchmark, as type-contingent realizations; see below). With perfect and free information, an optimal acceptance strategy is:
	\[
	\text{Accept deviation iff } S(\theta_f,\theta_\ell) \;\ge\; \Pi_{f\ell}.
	\]
	Therefore, the ex-ante expected joint surplus from ``deviation with test'' equals
	\[
	\mathbb{E}\!\left[\, S(\theta_f,\theta_\ell)\cdot \mathbf{1}\{S(\theta_f,\theta_\ell)\ge \Pi_{f\ell}\} \;+\; \Pi_{f\ell}\cdot \mathbf{1}\{S(\theta_f,\theta_\ell)< \Pi_{f\ell}\} \,\right],
	\]
	which can be written as
	\[
	\Pi_{f\ell} \;+\; \mathbb{E}\!\left[\,\big(S(\theta_f,\theta_\ell)-\Pi_{f\ell}\big)_+ \right],
	\quad\text{where } x_+=\max\{x,0\}.
	\]
	The ICPS blocking inequality with $(\pi,c)=(1,0)$ becomes
	\[
	\Pi_{f\ell} \;+\; \mathbb{E}\!\left[\,\big(S(\theta_f,\theta_\ell)-\Pi_{f\ell}\big)_+ \right]
	\;>\; \Pi_{f\ell},
	\]
	which holds iff
	\[
	\mathbb{P}\big( S(\theta_f,\theta_\ell) \;>\; \Pi_{f\ell} \big) \;>\; 0.
	\]
	Thus, $(\mu,p)$ fails ICPS-stability iff there is a \emph{positive-probability} set of type realizations on which $(f,\ell)$ can block given complete information.
	
	To conclude the equivalence, interpret $\Pi_{f\ell}$ as the pair’s (type-contingent) joint payoff under $(\mu,p)$ when information is complete. Then the condition ``for no pair $(f,\ell)$ does $S(\theta_f,\theta_\ell) > \Pi_{f\ell}$ on a positive-probability set'' is exactly the standard complete-information pairwise stability requirement (no realized-type blocking). Conversely, if complete-information pairwise stability holds, the indicator set $\{S(\theta_f,\theta_\ell)>\Pi_{f\ell}\}$ has probability zero for all pairs, so the expectation of the positive part is zero and no ICPS blocking deviation exists. Hence ICPS coincides with complete-information stability when a free perfect test is feasible.
	
	\medskip
	Combining (1)–(3) yields the proposition. 
\end{proof}

\noindent
Figures~\ref{fig:ipcs3d} and \ref{fig:ipcs} provide a graphical counterpart to 
mProposition~\ref{prop:refinement}. 
Panel~(a) displays the expected deviation surplus under ICPS as a function of the status quo joint payoff $\Pi$ (horizontal axis) and the mean of the surplus distribution (depth axis). 
The surface height corresponds to $\Pi+\mathbb{E}[(S-\Pi)_+]$, illustrating how credible information raises expected payoffs whenever $\mathbb{P}(S>\Pi)>0$. 
Panel~(b) provides the corresponding heatmap with contour lines, highlighting the regions of high option value of information. The figures clearly show that as $\Pi$ increases, the profitable set $\{S>\Pi\}$ shrinks and the expected ICPS deviation payoff converges to $\Pi$, aligning ICPS with Bayesian stability. 
Conversely, when $\Pi$ is low relative to the surplus distribution, the wedge between the baseline and ICPS expected payoff is substantial, reflecting the strengthening of the blocking condition under credible information. 
These patterns visually mirror the logical structure of parts~(2) and~(3) of Proposition~\ref{prop:refinement}.

\begin{figure}[htbp]
	\centering
	\begin{subfigure}[b]{0.48\textwidth}
		\centering
		\includegraphics[width=\textwidth]{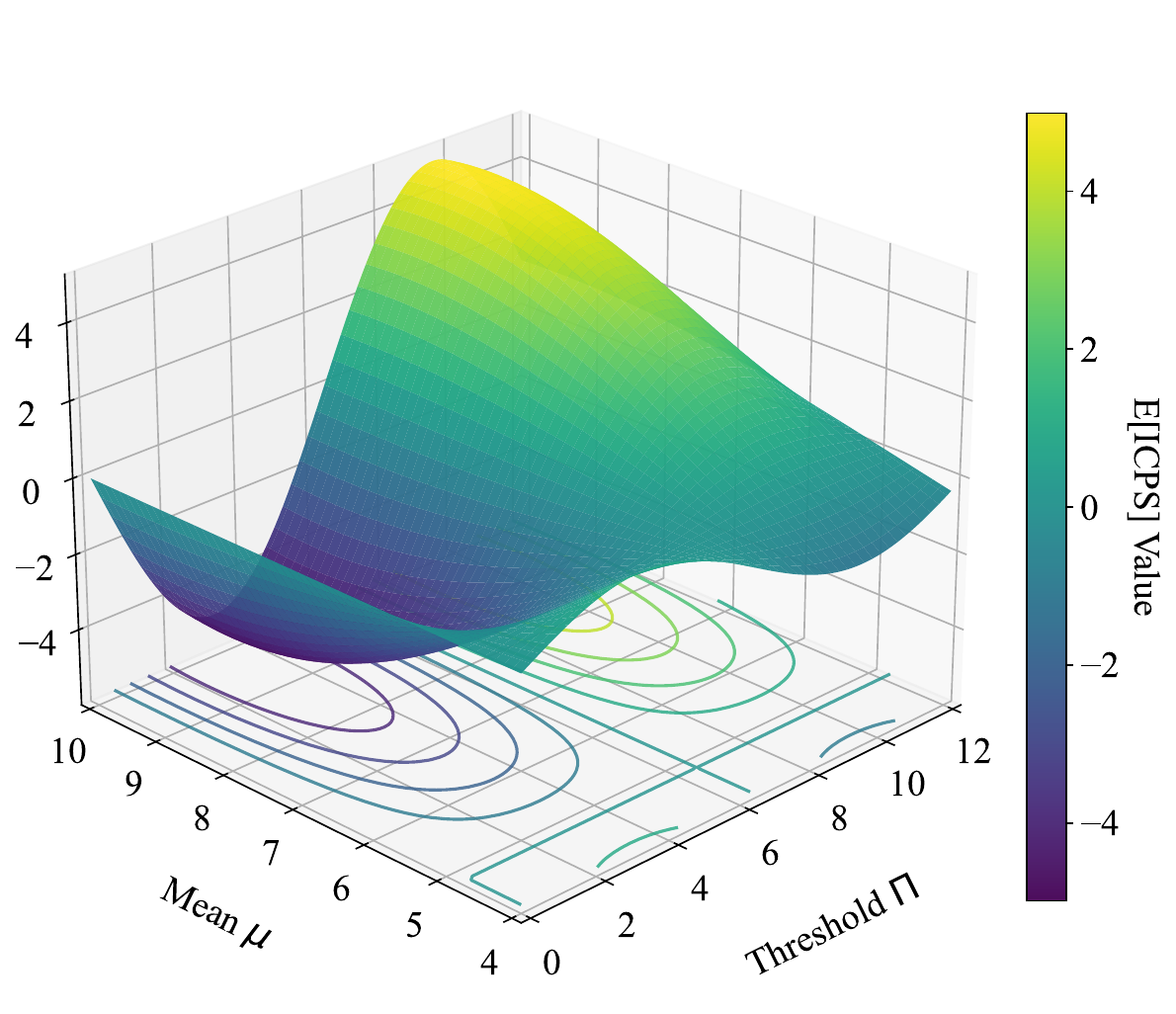}
		\caption{3D Surface with Contour Projection}
		\label{fig:ipcs3d}
	\end{subfigure}
	\hspace{1mm}
	\begin{subfigure}[b]{0.48\textwidth}
		\centering
		\includegraphics[width=\textwidth]{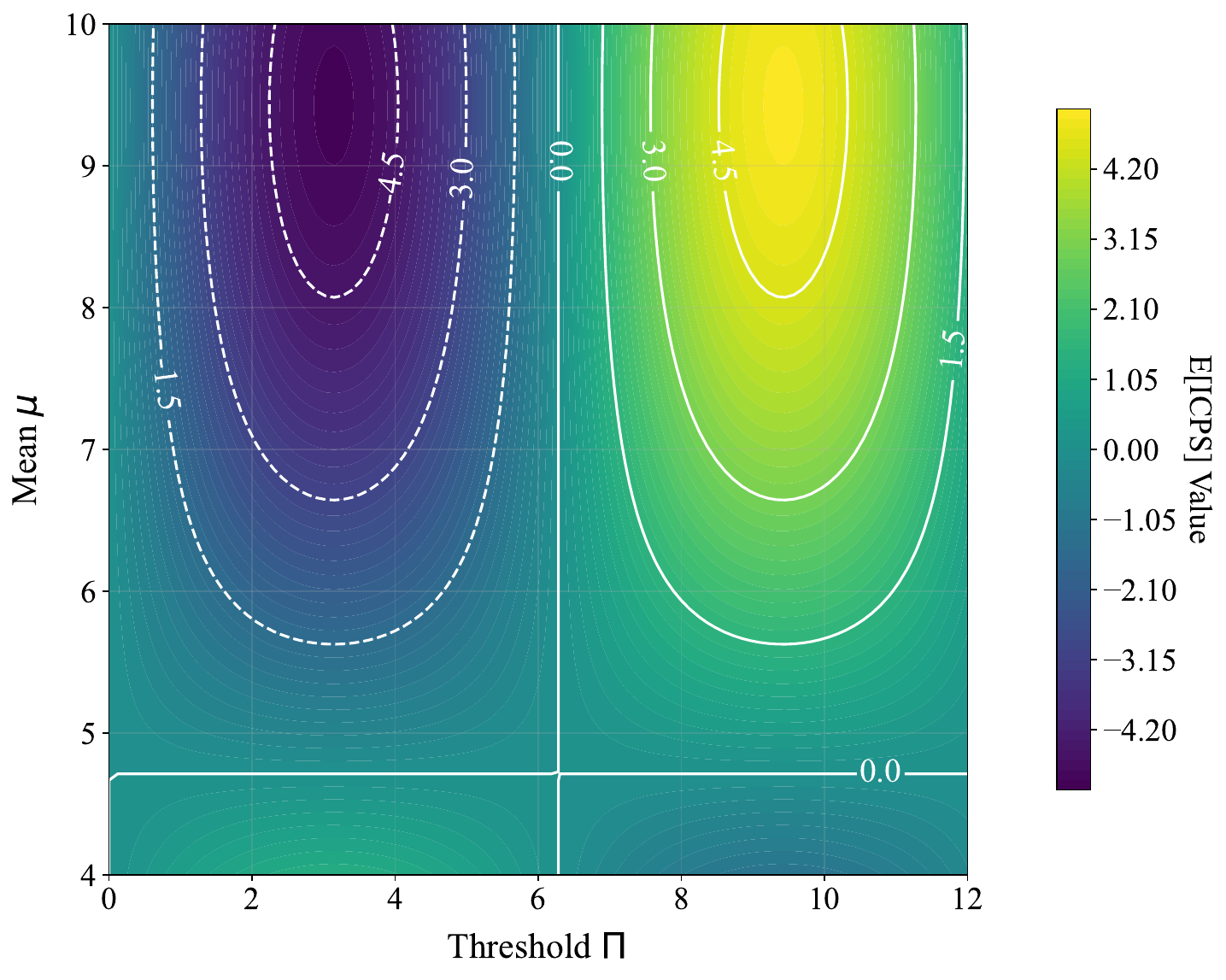}
		\caption{Heatmap with Contour Lines}
		\label{fig:ipcs}
	\end{subfigure}
	\caption{ICPS analysis: (a) 3D surface with contour projection and (b) heatmap with contour lines.}
	\label{fig:icps_combined}
\end{figure}

\section{MAIN RESULTS}\label{sec:mainresults}

This section establishes the central theoretical properties of Information-Credible Pairwise Stability (ICPS). I first show that ICPS refines Bayesian stability. I then characterize how credible information acquisition affects welfare, the structure of stable matchings, and the set of unmatched agents.

\subsection{Refinement and Surplus Decomposition}

\begin{proposition}[Refinement of Bayesian Stability]
	\label{prop:refinement_bayes}
	For any set of feasible information technologies $(\pi,c)$, the set of ICPS-stable allocations is a subset of the set of Bayesian-stable allocations:
	\[
	\Sigma^{\mathrm{ICPS}} \subseteq \Sigma^{\mathrm{Bayes}}.
	\]
\end{proposition}

\begin{proof}[Proof of Proposition~\ref{prop:refinement_bayes}](Refinement of Bayesian Stability)
	Let $(\mu,p)$ be ICPS-stable. Suppose, towards a contradiction, that $(\mu,p)$ is not Bayesian-stable. Then there exists a firm--worker pair $(f,\ell)$ and a transfer $t$ such that, under priors (no new information), deviating to match $(f,\ell)$ and splitting the expected surplus makes both weakly better off and at least one strictly better off. 
	
	But such a deviation is a special case of an ICPS deviation using a \emph{degenerate} test with accuracy $\pi=0$ (or, equivalently, choosing not to acquire information). Hence $(f,\ell)$ would also block $(\mu,p)$ under ICPS, contradicting ICPS stability. Therefore $\Sigma^{\mathrm{ICPS}} \subseteq \Sigma^{\mathrm{Bayes}}$. Figure~\ref{fig:venn_refinement} provides a graphical illustration of this inclusion relation, and also depicts the boundary case of equality when $\pi=0$ or $c=\infty$.
	
\end{proof}

\begin{figure}[htbp]
	\centering
	\begin{subfigure}[b]{0.4\textwidth}
		\centering
		\begin{tikzpicture}[scale=1.0]
			\draw[thick] (0,0) ellipse (3.2cm and 2.1cm);
			\node[anchor=west] at (3,1.4) {$\Sigma^{\mathrm{Bayes}}$};
			\draw[thick,fill=black!8] (-0.4,0) ellipse (1.9cm and 1.2cm);
			\node at (-0.4,-1.6) {$\Sigma^{\mathrm{ICPS}}$};
		\end{tikzpicture}
		\caption{\(\Sigma^{\mathrm{ICPS}} \subseteq \Sigma^{\mathrm{Bayes}}\)}
	\end{subfigure}\hfill
	\begin{subfigure}[b]{0.4\textwidth}
		\centering
		\begin{tikzpicture}[scale=1.0]
			\draw[thick,fill=black!8] (0,0) ellipse (3.2cm and 2.1cm);
			\node at (0,0.2) {$\Sigma^{\mathrm{ICPS}} = \Sigma^{\mathrm{Bayes}}$};
			\node at (0,-1.2) {\small (boundary case)};
		\end{tikzpicture}
		\caption{Uninformative/infeasible}
	\end{subfigure}
	\caption{ICPS-stable allocations refine Bayesian-stable ones.}
	\label{fig:venn_refinement}
\end{figure}
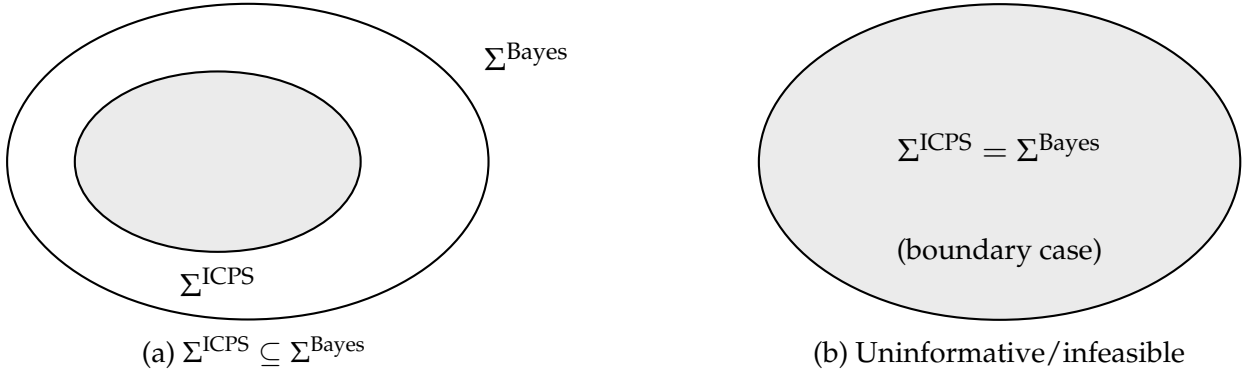

\begin{lemma}[Surplus Decomposition under ICPS]\label{lemma:surplus_decomposition}
	For any firm–worker pair $(f,\ell)$ and any test $(\pi,c)$,
	\[
	\mathbb{E}\!\big[S(\theta_f,\theta_\ell)\mid \text{test}\big]
	= \pi \,\mathbb{E}\!\big[S(\theta_f,\theta_\ell)\mid \text{perfect info}\big]
	+ (1-\pi)\,\mathbb{E}\!\big[S(\theta_f,\theta_\ell)\big].
	\]
	Let $A:=\{\text{the pair accepts/deviates after observing the test outcome}\}$ be the acceptance event
	induced by any (measurable) stopping rule based on the test outcome. Then the pair’s
	\emph{probability-weighted} expected surplus under that rule weakly exceeds the corresponding
	prior-based benchmark:
	\[
	\mathbb{E}\!\big[S(\theta_f,\theta_\ell)\mid A\big]\Pr(A)\;-\;c
	\;\;\ge\;\;
	\mathbb{E}\!\big[S(\theta_f,\theta_\ell)\big]\Pr(A)\;-\;c,
	\]
	with strict inequality whenever some accepted realization yields a posterior expected surplus strictly
	above the prior.
\end{lemma}

\begin{remark}[On “expected surplus” conventions]
	The inequality compares \emph{unconditional} expected payoffs: the conditional expectation 
	$\mathbb{E}[S\mid A]$ is multiplied by $\Pr(A)$ to reflect that the deviation occurs only on $A$.
	Equivalently, one can write $\mathbb{E}[S\,\mathbf{1}_A]-c \ge \mathbb{E}[S]\Pr(A)-c$.
\end{remark}

\begin{proof}[Proof of Lemma~\ref{lemma:surplus_decomposition}]
	Let $I$ denote the information revealed by the test about $(\theta_f,\theta_\ell)$, with
	$\Pr(\text{perfect info})=\pi$ and $\Pr(\text{no extra info})=1-\pi$. By the law of total expectation,
	\[
	\mathbb{E}\!\left[S\mid \text{test}\right]
	= \mathbb{E}\!\left[\mathbb{E}[S\mid I]\right]
	= \pi\,\mathbb{E}\!\left[S\mid \text{perfect info}\right]
	+ (1-\pi)\,\mathbb{E}[S].
	\]
	Let $A$ be the acceptance event implied by any stopping rule that consummates the deviation only on
	favorable realizations of $I$ (e.g., when the posterior exceeds a threshold). Then
	\[
	\mathbb{E}[S\,\mathbf{1}_A]
	= \mathbb{E}\!\left[\mathbb{E}[S\mid I]\,\mathbf{1}_A\right]
	\;\ge\;
	\mathbb{E}[S]\,\Pr(A),
	\]
	with strict inequality when the rule accepts on some realizations with $\mathbb{E}[S\mid I]>\mathbb{E}[S]$.
	Subtracting the (deterministic) test cost $c$ and writing
	$\mathbb{E}[S\,\mathbf{1}_A]=\mathbb{E}[S\mid A]\Pr(A)$ gives
	\[
	\mathbb{E}[S\mid A]\Pr(A)-c \;\ge\; \mathbb{E}[S]\Pr(A)-c,
	\]
	and strict inequality under the stated condition. This proves both the decomposition and the
	probability-weighted weak-improvement claim.
\end{proof}

\subsection{Welfare Comparison}

\begin{proposition}[Welfare Improvement under ICPS]
	\label{prop:welfare_improvement}
	Under Assumptions~1--3, any ICPS-stable allocation $(\mu^{\mathrm{ICPS}},p)$ generates weakly higher expected total surplus than any Bayesian-stable allocation $(\mu^{\mathrm{Bayes}},p)$:
	\[
	\mathbb{E}\big[W(\mu^{\mathrm{ICPS}},p)\big]
	\;\geq\;
	\mathbb{E}\big[W(\mu^{\mathrm{Bayes}},p)\big].
	\]
	The inequality is strict whenever $\pi>0$, $c$ is sufficiently small, and there exists at least one inefficient match (e.g., $(H,L)$ or $(L,L)$) under $\mu^{\mathrm{Bayes}}$.
\end{proposition}

\begin{proof}[Proof of Proposition~\ref{prop:welfare_improvement}](Welfare Improvement under ICPS)
	Take any Bayesian-stable allocation $(\mu^B,p^B)$. If no ICPS-blocking pair exists, then $(\mu^B,p^B)$ is also ICPS-stable and the result holds trivially. Otherwise, there exists a pair $(f,\ell)$ and a feasible test $(\pi,c)$ such that the posterior-conditional deviation increases the pair’s joint payoff strictly above $\pi_f(\mu^B,p^B)+\pi_\ell(\mu^B,p^B)$ (net of $c$), by Definition of ICPS blocking and Lemma~\ref{lemma:surplus_decomposition}. Because transfers are allowed, I may implement the deviation and redistribute the (strictly) higher expected match surplus between $f$ and $\ell$ without harming others; in particular, the deviation raises \emph{total} expected surplus $W$ strictly (the change in total expected surplus equals the increase in the deviators’ joint expected payoff).
	
	Define an improvement path that repeatedly executes any ICPS-blocking deviation, updating the allocation each time. Each step raises $W$ strictly by at least some $\varepsilon>0$ (bounded away from zero given finiteness of the relevant partitions/types and discrete feasible-test set). Since (i) the set of matchings on finite $F,L$ is finite, and (ii) $W$ is uniformly bounded above by $\sum_{f\in F}\max_{\theta_f,\theta_\ell} S(\theta_f,\theta_\ell)$, the process must terminate in finitely many steps at an allocation $(\mu^{\mathrm{ICPS}},p^{\mathrm{ICPS}})$ with no ICPS-blocking pair, i.e., an ICPS-stable allocation. Summing the strict improvements along the path gives
	\[
	\mathbb{E}\!\left[W(\mu^{\mathrm{ICPS}},p^{\mathrm{ICPS}})\right]
	\;>\;
	\mathbb{E}\!\left[W(\mu^{B},p^{B})\right]
	\]
	whenever an ICPS-blocking pair existed at $(\mu^B,p^B)$, and weak inequality otherwise. This establishes the claim and the strictness condition in the statement.
\end{proof}

\subsubsection{Sorting Structure}

\begin{lemma}[Sorting Property]
	\label{lemma:sorting}
	Under surplus monotonicity, no high-type agent is matched with a low-type agent in any ICPS-stable allocation when testing is perfect and costless. More generally, ICPS-stable matchings are weakly positively assortative.
\end{lemma}

\begin{proof}[Proof of Lemma~\ref{lemma:sorting}](Sorting Property)
	Assume $c=0$ and $\pi=1$ (costless, perfect testing). Suppose, for contradiction, that some ICPS-stable allocation $(\mu,p)$ contains an $H$--$L$ match, say $(f_H,\ell_L)$. Consider the other side’s set of partners. If there exists a worker $\tilde\ell_H$ of high type, then pair $(f_H,\tilde\ell_H)$ can acquire (perfect) information at zero cost, verify both are $H$, and deviate to obtain surplus $S(H,H)=\alpha$. Since $S(H,L)=\beta$ and $\alpha>\beta$ by Assumption~1, the deviating pair can split the surplus improvement via transfers, forming a profitable ICPS deviation that contradicts stability.
	
	If instead there is no high-type worker available (so every worker is low), then any firm--worker match is either $H$--$L$ or $L$--$L$. Pick any low-type firm $f_L$ matched with a low-type worker (or unmatched). Then $(f_H,\ell_L)$ yields $\beta$ while $(f_L,\ell_L)$ yields $\gamma$, with $\beta>\gamma$; swapping partners cannot eliminate a profitable pairwise deviation for some $H$--$H$ comparison on the firms’ side unless there are no high types on that side either. Symmetrically, if there exists a high-type firm and a high-type worker across sides, they create a profitable $H$--$H$ deviation at zero cost. Hence in any ICPS-stable allocation under $c=0,\pi=1$, no $H$--$L$ matches can persist whenever $H$ agents exist on both sides. Therefore ICPS-stable allocations are (weakly) positively assortative; the result extends by continuity to sufficiently high $\pi$ and sufficiently small $c$, since a small (net) information friction preserves the strict ranking $\alpha>\beta>\gamma$ in expected terms.
\end{proof}

\subsection{Existence of ICPS-Stable Allocations}

\begin{proposition}[Existence]
	\label{prop:existence}
	Under Assumptions~1--3 and finite agent sets,
	\[
	\Sigma^{\mathrm{ICPS}} \neq \emptyset.
	\]
\end{proposition}

	\begin{proof}[Proof of Proposition~\ref{prop:existence}]
	Let $W(\mu) \equiv \sum_{f \in F} \mathbb{E}\big[S(\theta_f, \theta_{\mu(f)})\big]$ denote total expected surplus. Since $F$ and $L$ are finite, the set of matchings is finite, so there exists $\mu^\star \in \arg\max_{\mu} W(\mu)$. Choose transfers $p^\star$ arbitrarily (TU guarantees existence of supporting prices).
	
	Suppose, for contradiction, that $(\mu^\star, p^\star)$ admits an ICPS-blocking pair. Then for some feasible test $(\pi, c)$ and acceptance rule, the deviating pair's posterior-conditional deviation yields strictly higher joint expected payoff than their status-quo sum. By Lemma~\ref{lemma:surplus_decomposition} and transferable utility, this implies the existence of a matching $\mu'$ with $W(\mu') > W(\mu^\star)$, contradicting the maximality of $\mu^\star$. Hence $(\mu^\star, p^\star)$ is ICPS-stable and $\Sigma^{\mathrm{ICPS}} \neq \emptyset$.
	
	\smallskip
	\noindent\emph{Remark.} If test costs represent resource costs, replace $W$ with $\Phi(\mu) = W(\mu) - \mathbb{E}[\text{total test costs}]$ and pick $\mu^\star \in \arg\max_\mu \Phi(\mu)$; the contradiction argument applies identically.
\end{proof}

\subsection{Invariance of the Unmatched Set}

\begin{proposition}[Lone Wolf Property under ICPS]
	\label{prop:lonewolf}
	If agent $i$ is unmatched in some ICPS-stable allocation $(\mu,p)$, then $i$ is unmatched in all ICPS-stable allocations:
	\[
	\mu(i) = \emptyset \ \text{in some } (\mu,p) \in \Sigma^{\mathrm{ICPS}}
	\ \implies \
	\mu'(i) = \emptyset \ \text{for all } (\mu',p') \in \Sigma^{\mathrm{ICPS}}.
	\]
\end{proposition}

\begin{proof}[Proof of Proposition~\ref{prop:lonewolf}]
	Assume Assumptions~\ref{ass:mono}--\ref{ass:pos} and that types are drawn from finite sets with strict surplus monotonicity. 
	Suppose, for contradiction, there exists an agent $i$ (w.l.o.g., a worker) such that 
	$i$ is unmatched in some ICPS-stable allocation $(\mu,p)$ but matched in another ICPS-stable allocation $(\mu',p')$.
	
	Let $j = \mu'(i)$ denote $i$’s match in $(\mu',p')$. 
	Because $i$ is unmatched in $(\mu,p)$, $\pi_i(\mu,p)=0$. 
	Let $V_{\text{dev}}(i,j)$ denote the maximum joint expected deviation value achievable by $(i,j)$ using any feasible test $(\pi,c)$:
	\[
	V_{\text{dev}}(i,j) \;=\; \max_{(\pi,c)\in \mathcal{T}} \big\{ \mathbb{E}[S(\theta_i,\theta_j)\mid \text{test}] - c \big\}.
	\]
	In $(\mu',p')$, the total surplus generated by the pair $(i,j)$ is bounded above by $V_{\text{dev}}(i,j)$, since $(i,j)$ cannot obtain more expected surplus than what they could get under their best feasible test and transfer arrangement.
	
	ICPS stability of $(\mu',p')$ implies
	\[
	\pi_i(\mu',p') + \pi_j(\mu',p') \;\leq\; V_{\text{dev}}(i,j),
	\]
	because their joint payoff cannot exceed what they could obtain by deviating together.
	
	Now consider $(\mu,p)$. 
	If $\pi_j(\mu,p) < V_{\text{dev}}(i,j)$, then $(i,j)$ can profitably deviate from $(\mu,p)$: 
	$i$ improves from $0$ and $j$ obtains at least the same or higher joint surplus compared to its current match, 
	contradicting ICPS stability of $(\mu,p)$.
	
	If $\pi_j(\mu,p) \geq V_{\text{dev}}(i,j)$, then $j$ must already be receiving at least the maximal surplus obtainable from any match with $i$ under any test. 
	But this implies that $j$ cannot be strictly better off matching with $i$ in $(\mu',p')$, and since $i$ receives at most $0$ in $(\mu,p)$, $(i,j)$ cannot form part of an alternative ICPS-stable allocation with higher joint surplus. 
	This contradicts the assumption that they are matched in $(\mu',p')$.
	
	Therefore, no agent $i$ can be unmatched in one ICPS-stable allocation and matched in another. 
\end{proof}

\subsection{Uniqueness}

\begin{proposition}[Uniqueness of ICPS-Stable Matching]\label{prop:uniqueness}
	Suppose Assumptions~\ref{ass:mono}--\ref{ass:pos} hold, types are distinct (no ties), and 
	Assumption~\ref{ass:testpower} holds. Then the ICPS-stable matching is unique and equals the positively assortative matching.
\end{proposition}

\begin{proof}[Proof of Proposition~\ref{prop:uniqueness}]
	Suppose $(\mu,p)$ and $(\tilde\mu,\tilde p)$ are distinct ICPS-stable allocations. 
	Let $(f^*,\ell^*)$ be the pair with the highest type product order not matched together in at least one of them; w.l.o.g.\ they are not matched in $(\mu,p)$. 
	Then in $(\mu,p)$ at least one of $f^*$ or $\ell^*$ is matched with a strictly lower-type partner. 
	Let $S_{\mathrm{worse}} \equiv \max\{S(f^*,\ell'),\,S(f',\ell^*)\}$ denote the larger of the two “degraded” surpluses; by definition of $\Delta$,
	\[
	S(f^*,\ell^*) - S_{\mathrm{worse}} \;\ge\; \Delta.
	\]
	Consider the feasible test $(\bar\pi,\bar c)$ from Assumption~\ref{ass:testpower}. 
	If $(f^*,\ell^*)$ deviate and accept only on favorable signals, their expected net value is at least
	\[
	\begin{aligned}
		&\bar\pi\,S(f^*,\ell^*) + (1-\bar\pi)\,S_{\mathrm{worse}} - \bar c \\
		&= S_{\mathrm{worse}} + \bar\pi\big(S(f^*,\ell^*)-S_{\mathrm{worse}}\big) - \bar c \ge S_{\mathrm{worse}} + \bar\pi\,\Delta - \bar c > S_{\mathrm{worse}}.
	\end{aligned}
	\]
	Because transfers are unrestricted, the current joint payoff of $f^*$ and $\ell^*$ under $(\mu,p)$ cannot exceed $S_{\mathrm{worse}}$ (they are split across strictly inferior pairings). 
	Hence $(f^*,\ell^*)$ form an ICPS-blocking pair against $(\mu,p)$, contradicting its ICPS stability. 
	Therefore no such discrepancy exists, and the unique ICPS-stable matching is the positively assortative one.
\end{proof}
\begin{proposition}[Uniqueness under Perfect Information]\label{prop:uniqueness_perfect}
	If $(\pi,c)=(1,0)$ is feasible, Assumptions~\ref{ass:mono}--\ref{ass:pos} hold, and types are distinct, 
	then the (ICPS-)stable matching is unique and equals the positively assortative matching.
\end{proposition}

\begin{proof}[Proof of Proposition~\ref{prop:uniqueness_perfect}]
	With perfect, costless information, ICPS coincides with complete-information stability. 
	Under strict monotonicity and distinct types, the assortative matching uniquely maximizes pairwise surplus; any allocation not matching $(f,\ell)$ with their sorted counterpart admits a profitable deviation by that pair (no test cost, perfect revelation), contradicting stability. 
	Thus the stable matching is unique and positively assortative.
\end{proof}

\section{EXTENSIONS}\label{sec:extensions}

This section explores two natural extensions of the ICPS framework. 
First, I examine how the stability concept behaves in the absence of transfers (i.e., in non-transferable utility (NTU) environments).
Second, I relax the independence assumption on types and study the effect of interdependent (correlated) types on ICPS refinement.

\subsection{Absence of Transfers: NTU-ICPS Stability}

So far, the analysis has been conducted in a transferable utility (TU) environment, where any increase in joint surplus can be redistributed between agents through transfers.
I now consider an NTU environment, in which each agent's utility depends directly on the realized match surplus and cannot be altered through side payments.

Let $u_f(\theta_f,\theta_\ell)$ and $u_\ell(\theta_f,\theta_\ell)$ denote the utilities of the firm and worker from matching $(f,\ell)$ when types $(\theta_f,\theta_\ell)$ are realized.
The current expected payoffs under allocation $\mu$ are $\pi_f(\mu)$ and $\pi_\ell(\mu)$.

\begin{definition}[NTU Information-Credible Pairwise Stability]
	An allocation $\mu$ is \emph{NTU-ICPS-stable} if it is individually rational and there exists no pair $(f,\ell)$ and feasible test $(\pi,c)$ such that
	\begin{align}
		\E\big[u_f(\theta_f,\theta_\ell)\mid \text{test}\big] - c_f &> \pi_f(\mu), \\
		\E\big[u_\ell(\theta_f,\theta_\ell)\mid \text{test}\big] - c_\ell &> \pi_\ell(\mu),
	\end{align}
	where $c_f$ and $c_\ell$ are the individual testing costs borne by the firm and worker, respectively.
\end{definition}

\begin{remark}
	In the absence of transfers, a deviation is profitable only if it improves both sides’ expected utilities simultaneously.
	Hence, information acquisition may be individually unprofitable even if it raises joint surplus.
\end{remark}

\begin{proposition}[NTU Refinement]\label{prop:ntu_refinement}
	For any feasible set of information acquisition technologies $(\pi,c)$,
	\[
	\Sigma^{\mathrm{ICPS, NTU}} \;\subseteq\; \Sigma^{\mathrm{Bayes, NTU}}.
	\]
	The refinement is strict only when information acquisition raises both agents’ expected utilities simultaneously.
\end{proposition}

\begin{proof}
	Suppose $\mu$ is NTU-ICPS-stable but not NTU-Bayesian-stable. 
	Then there exists a pair $(f,\ell)$ such that
	\[
	\E\big[u_f(\theta_f,\theta_\ell)\big] > \pi_f(\mu)
	\quad \text{and} \quad
	\E\big[u_\ell(\theta_f,\theta_\ell)\big] > \pi_\ell(\mu),
	\]
	evaluated under priors. 
	By Remark~\ref{rem:nulltest}, the pair can choose the null test $(\pi,c)=(0,0)$ and consummate deviation without information.
	This constitutes an NTU-ICPS blocking deviation, contradicting NTU-ICPS stability. 
	Hence $\Sigma^{\mathrm{ICPS, NTU}} \subseteq \Sigma^{\mathrm{Bayes, NTU}}$.
	
	Strict refinement arises only if credible testing raises \emph{both} sides’ expected utilities, since without transfers the gains must be Pareto-improving for both agents individually. 
\end{proof}

\begin{remark}
	Unlike the TU case, NTU-ICPS does not guarantee welfare improvement relative to Bayesian stability.
	Information may increase joint surplus but fail to induce a deviation because the surplus cannot be redistributed.
\end{remark}

\begin{remark}[Relationship to Matching with Contracts]
	The NTU-ICPS concept can be viewed through the lens of \citet{hatfield2005matching}'s matching-with-contracts framework, 
	where each potential match corresponds to a contract specifying fixed payoffs. 
	ICPS then represents stability with respect to a richer set of contingent contracts enabled by information acquisition.
\end{remark}

\begin{example}[Information Without Blocking in NTU]
	Consider a pair $(f,\ell)$ with current payoffs $(4,4)$. 
	Suppose testing reveals that with probability $0.5$, types are $(H,H)$ yielding utilities $(8,2)$, 
	and with probability $0.5$, types are $(L,L)$ yielding $(2,8)$. 
	The expected post-test utilities are $(5,5)$, increasing joint surplus from $8$ to $10$. 
	However, under NTU, neither agent would agree to test and deviate, 
	as each faces a $50\%$ chance of being worse off. 
	Thus, the allocation remains NTU-ICPS-stable despite potential efficiency gains.
\end{example}

\subsection{Interdependent Types}\label{subsec:correlated}

I now relax Assumption~\ref{ass:indep} (independence of types) and allow for correlation between firm and worker types. 
Let $\rho = \mathrm{corr}(\theta_f,\theta_\ell)\in[-1,1]$.
Then
\[
\Prob(\theta_f,\theta_\ell) \;\neq\; \Prob(\theta_f)\Prob(\theta_\ell),
\]
and credible testing alters \emph{joint} beliefs about both sides’ types.

\begin{assumption}[Correlated Types Structure]\label{ass:correlated}
	Types $(\theta_f,\theta_\ell)$ follow a joint distribution with correlation $\rho$, 
	where $\Prob(\theta_f=H,\theta_\ell=H) = p^2 + \rho p(1-p)$,
	$\Prob(\theta_f=H,\theta_\ell=L) = p(1-p)(1-\rho)$, and so on,
	ensuring marginal probabilities remain $\Prob(\theta_i=H)=p$ for $i\in\{f,\ell\}$.
\end{assumption}

\begin{lemma}[Effect of Correlation on ICPS Refinement]\label{lemma:correlation}
	Let $\rho$ denote the correlation between firm and worker types.
	\begin{enumerate}
		\item If $\rho>0$, the ICPS refinement relative to Bayesian stability vanishes in the limit as $\rho\to 1$:
		\[
		\lim_{\rho \to 1}\Sigma^{\mathrm{ICPS}} = \Sigma^{\mathrm{Bayes}}.
		\]
		\item If $\rho=0$, the model reduces to the baseline independent-type case.
		\item If $\rho<0$ and $(\pi,c)$ is informative and feasible, ICPS strictly refines Bayesian stability:
		\[
		\Sigma^{\mathrm{ICPS}} \subset \Sigma^{\mathrm{Bayes}}.
		\]
	\end{enumerate}
\end{lemma}

\begin{proposition}[Quantitative Refinement under Correlation]\label{prop:quant_refinement}
	The magnitude of ICPS refinement, measured by $|\Sigma^{\mathrm{Bayes}} \setminus \Sigma^{\mathrm{ICPS}}|$, 
	is monotonic in $|\rho|$ when $\rho<0$ and decreasing in $\rho$ when $\rho>0$.
	Specifically, for $\rho<0$,
	\[
	\frac{\partial}{\partial \rho}\Big|\Sigma^{\mathrm{Bayes}} \setminus \Sigma^{\mathrm{ICPS}}\Big| < 0.
	\]
\end{proposition}

\begin{proof}
	(i) When $\rho>0$, observing one’s own type provides strong information about the partner’s type. 
	Bayesian priors already internalize this correlation, so the expected posterior after testing converges to the prior as $\rho\to 1$.
	Hence the option value of information vanishes, and the ICPS and Bayesian blocking sets coincide in the limit.
	
	(ii) When $\rho=0$, the types are independent and the original ICPS–Bayes relationship (Proposition~\ref{prop:refinement}) holds.
	
	(iii) When $\rho<0$, high-type firms are ex ante more likely to be matched with low-type workers and vice versa. 
	An informative test allows deviating pairs to condition on rare high–high realizations with high posterior probability but low prior probability. 
	This strictly raises expected deviation payoffs whenever $\pi>0$ and $c$ is sufficiently small.
	Therefore, there exist ICPS-blocking deviations that are not Bayesian-blocking deviations, and the refinement is strict.
	Monotonicity follows from the fact that the posterior–prior gap grows with the strength of negative correlation.
\end{proof}

\begin{remark}
	Correlation between types shifts the information frontier: 
	positive correlation makes priors more informative, reducing the marginal impact of credible testing,
	while negative correlation amplifies the option value of information and strengthens ICPS refinement.
\end{remark}

\begin{remark}[Economic Interpretation of Correlation Effects]
	Positive correlation ($\rho>0$) arises in assortative industries where high-quality firms naturally match with high-quality workers. 
	Negative correlation ($\rho<0$) may arise in compensatory matching environments, 
	where firms seek partners with complementary strengths. 
	ICPS has greater impact in the latter case by revealing unexpected complementarities.
\end{remark}

\begin{corollary}[Existence and Uniqueness Under Correlation]\label{cor:existence_corr}
	Propositions~\ref{prop:existence} and \ref{prop:uniqueness} continue to hold under correlated types, 
	since the refinement structure depends on the informativeness of tests and surplus monotonicity, not on independence per se.
	However, the \emph{magnitude} of the refinement depends on $\rho$.
\end{corollary}

\subsection{Endogenous Test Selection}\label{subsec:endogenous}

In the baseline ICPS formulation, deviating pairs have access to a single exogenously given test $(\pi,c)$.
I now extend the framework to allow \emph{endogenous test selection}: 
agents choose which test to use from a menu of available technologies.

Let $\mathcal{T} = \{(\pi_i,c_i)\}_{i=1}^n$ denote a finite menu of information acquisition technologies, 
where each pair $(\pi_i,c_i)$ represents an accuracy–cost bundle. 
I assume that $\pi_i \in [0,1]$ and $c_i \ge 0$ for all $i$, 
and that $\mathcal{T}$ contains the null test $(\pi,c)=(0,0)$.

Given allocation $(\mu,p)$, the expected joint payoff from deviation using technology $i$ is
\[
V_{f\ell}(i) \;=\; \E\big[S(\theta_f,\theta_\ell)\mid \text{test } i\big] - c_i.
\]
The ICPS blocking condition becomes
\begin{equation}\label{eq:endogenous_blocking}
	\max_{i\in \mathcal{T}} V_{f\ell}(i) \;>\; \pi_f(\mu) + \pi_\ell(\mu).
\end{equation}
Thus, deviating pairs endogenously choose the test that maximizes their expected net surplus prior to deciding on deviation.

\begin{remark}
	Endogenous test selection allows agents to exploit the most cost-effective information acquisition technology. 
	Compared to the baseline, this weakly enlarges the set of profitable deviations, leading to an additional refinement of Bayesian stability.
\end{remark}

\begin{proposition}[Refinement with Endogenous Test Selection]\label{prop:endogenous_refinement}
	Let $\mathcal{T}$ be the menu of available information acquisition technologies.
	Then
	\[
	\Sigma^{\mathrm{ICPS, Endog}} \;\subseteq\; \Sigma^{\mathrm{ICPS}} \;\subseteq\; \Sigma^{\mathrm{Bayes}}.
	\]
	If at least one technology $i$ satisfies $\pi_i>0$ and $c_i$ sufficiently small, then the inclusion is strict.
\end{proposition}

\begin{proof}
	For any pair $(f,\ell)$, the baseline ICPS blocking condition corresponds to using a single exogenous test $(\bar\pi,\bar c)$.
	Under endogenous selection, the pair can always replicate this choice by selecting the same $(\bar\pi,\bar c)$ from $\mathcal{T}$.
	Therefore, $\Sigma^{\mathrm{ICPS, Endog}} \subseteq \Sigma^{\mathrm{ICPS}}$.
	
	If the menu $\mathcal{T}$ contains a test $(\pi_i,c_i)$ such that
	\[
	\E[S\mid \text{test } i] - c_i > \E[S\mid \text{test } \bar{\pi},\bar{c}],
	\]
	then some pairs will block under endogenous selection but not under the baseline.
	Hence the refinement is strict whenever there exists at least one strictly more informative or cheaper technology.
	
	The second inclusion $\Sigma^{\mathrm{ICPS}}\subseteq\Sigma^{\mathrm{Bayes}}$ follows from Proposition~\ref{prop:refinement}.
\end{proof}

\begin{remark}[Welfare Implications]
	Endogenous selection enhances \emph{efficiency at the margin}: 
	deviating pairs optimally choose the most beneficial test,
	raising expected joint surplus relative to exogenous testing. 
	However, equilibrium allocations may not achieve first-best welfare if testing costs are heterogeneous or if only a subset of pairs has access to high-quality technologies.
\end{remark}

\begin{remark}[Strategic Heterogeneity]
	If the menu $\mathcal{T}$ differs across agents or if test costs are agent-specific, 
	test choice itself becomes a strategic decision that interacts with matching patterns.
	For example, high-type firms may invest in more accurate (costly) tests than low-type firms, 
	leading to endogenous information asymmetries and stratification in the matching market.
\end{remark}

\subsection{Sequential Testing and Real Options}\label{subsec:sequential}

I now allow firms and workers to acquire information \emph{sequentially} rather than in a single, lump-sum test. 
Sequential testing captures situations where information is obtained gradually—e.g., through interviews, trial contracts, or staged screening—so that agents can decide to \emph{abort} the deviation after receiving partial signals.

Let $\mathcal{T}^{\mathrm{seq}}$ denote a sequence of $K$ test stages:
\[
\mathcal{T}^{\mathrm{seq}} = \big\{ (\pi_1,c_1), (\pi_2,c_2), \dots, (\pi_K,c_K) \big\},
\]
where $0 < \pi_1 < \pi_2 < \dots < \pi_K \le 1$ and costs satisfy $0 < c_1 < c_2 < \dots < c_K$.
After each stage $k$, the pair observes a signal with accuracy $\pi_k$ and decides whether to:
\begin{enumerate}
	\item continue testing (paying additional cost $c_{k+1}$),
	\item abort deviation and stay with the current match,
	\item or consummate deviation given current information.
\end{enumerate}

The expected deviation value under sequential testing is
\[
V_{f\ell}^{\mathrm{seq}}
\;=\;
\max_{\tau \in \mathcal{S}}
\E\big[ S(\theta_f,\theta_\ell) \cdot \mathbf{1}_{\{\text{accept at stopping time }\tau\}} - C(\tau) \big],
\]
where $\tau$ is a stopping rule in the set $\mathcal{S}$ of feasible stopping strategies, 
and $C(\tau)$ is the cumulative testing cost incurred up to $\tau$.

\begin{remark}[Real Option Value]
	Sequential testing creates an \emph{option value of waiting}:
	agents can condition their decision on early signals, incurring further costs only if signals are favorable. 
	This weakly dominates one-shot testing in expected value terms.
\end{remark}

\begin{proposition}[Refinement with Sequential Testing]\label{prop:seq_refinement}
	Let $\Sigma^{\mathrm{ICPS, Seq}}$ denote the set of ICPS-stable allocations under sequential testing. Then
	\[
	\Sigma^{\mathrm{ICPS, Seq}} \;\subseteq\; \Sigma^{\mathrm{ICPS, Endog}} \;\subseteq\; \Sigma^{\mathrm{ICPS}} \;\subseteq\; \Sigma^{\mathrm{Bayes}}.
	\]
	The inclusion is strict whenever sequential testing generates positive option value, i.e.,
	\[
	V_{f\ell}^{\mathrm{seq}} > \max_{i\in\mathcal{T}} V_{f\ell}(i)
	\]
	for some pair $(f,\ell)$.
\end{proposition}

\begin{proof}
	Under sequential testing, agents can always replicate the outcome of any one-shot test by choosing an appropriate stopping rule $\tau$ at the corresponding stage.
	Hence $\Sigma^{\mathrm{ICPS, Seq}} \subseteq \Sigma^{\mathrm{ICPS, Endog}}$.
	Strict inclusion holds whenever sequential testing allows agents to screen out unfavorable states early and avoid full testing costs, thereby increasing expected deviation surplus.
	
	The remaining inclusions follow from Proposition~\ref{prop:endogenous_refinement} and Proposition~\ref{prop:refinement}.
\end{proof}

\begin{remark}[Efficiency and Selection Effects]
	Sequential testing refines the set of stable allocations by making deviation profitable for some pairs that could not profitably deviate under one-shot testing.
	This may induce better assortative matching and higher expected welfare.
	However, it can also generate \emph{information-induced sorting}: 
	agents with higher type priors are more likely to find it optimal to continue testing,
	leading to endogenous stratification in test intensity and match quality.
\end{remark}

\begin{remark}[Connection to Option Pricing Logic]
	This extension is conceptually related to \citet{dixit1994investment}'s real options framework:
	sequential testing embeds an \emph{irreversible cost structure} with the option to abandon.
	This makes the ICPS refinement more robust to uncertainty and allows agents to flexibly adjust their information acquisition strategy ex post.
\end{remark}

\begin{corollary}[Existence under Endogenous and Sequential Testing]
	Under Assumptions~\ref{ass:mono}--\ref{ass:pos} and finite agent sets,
	\[
	\Sigma^{\mathrm{ICPS, Seq}} \neq \emptyset.
	\]
	Moreover, the unique ICPS-stable matching under perfect and costless sequential testing remains the positively assortative matching, as in Proposition~\ref{prop:uniqueness_perfect}.
\end{corollary}

\section{CONCLUSION}\label{sec:con}

This paper introduces \emph{Information-Credible Pairwise Stability} (ICPS), a refinement of Bayesian stability for matching markets with incomplete information. ICPS incorporates credible bilateral information acquisition into the blocking condition, shifting the source of stability from fear-driven beliefs to credible opportunities. I show that ICPS strictly refines Bayesian stability and converges to complete-information stability when tests are perfect and free. ICPS-stable allocations always exist, generate weakly higher welfare, and are positively assortative. Under sufficient test power, the stable matching is unique. The framework accommodates extensions such as correlated types, non-transferable utility, and endogenous testing. More broadly, ICPS provides a foundation for analyzing matching markets in information-rich environments.

	\pagebreak




\onehalfspacing
\bibliographystyle{apacite} 
\bibliography{matching}

\end{document}